\documentclass[12pt]{article}
\usepackage{graphicx}
\usepackage{amsfonts,amsmath}
\usepackage[mathscr]{eucal}
\usepackage{amssymb}
\usepackage{amsthm}
\usepackage{bbold}
\theoremstyle{plain}
\newtheorem{thm}{Theorem}

\newtheorem{lem} {Lemma}

\newcommand{\be}{\begin{eqnarray}}
\newcommand{\ee}{\end{eqnarray}}
\newcommand{\bc}{\begin{center}}
\newcommand{\ec}{\end{center}}
\newcommand{\nn}{\nonumber \\}
\newcommand{\lb}{\label}
\newcommand{\p}[1]{(\ref{#1})}

\begin{document}

\begin{titlepage}

\vspace*{0.2cm}

\renewcommand{\thefootnote}{\star}
\begin{center}

{\LARGE\bf  Witten index for weak supersymmetric systems: invariance under deformations}

\vspace{2cm}

{\Large Andrei Smilga} \\

\vspace{0.5cm}

{\it SUBATECH, Universit\'e de
Nantes,  4 rue Alfred Kastler, BP 20722, Nantes  44307, France. }

\end{center}
\vspace{0.2cm} \vskip 0.6truecm \nopagebreak  

   \begin{abstract}
\noindent  

When a $4D$ supersymmetric theory is placed on $S^3 \times \mathbb{R}$, the supersymmetric algebra is necessarily modified to $su(2|1)$ and we are dealing with a {\it weak} supersymmetric system. For such systems, the excited states of the Hamiltonian are not all paired. As a result, the Witten index Tr$\{(-1)^F e^{-\beta H}\}$ is no longer an integer number, but a $\beta$-dependent function. 

 However, this function stays invariant under deformations of the theory that keep the supersymmetry algebra intact. Based on the Hilbert space analysis, we give a simple general proof of this fact. We then show how this invariance works for two simplest weak supersymmetric quantum mechanical systems involving a real or a complex bosonic degree of freedom.

   \end{abstract}

\end{titlepage}

\setcounter{footnote}{0}

\setcounter{equation}0

\section{Introduction} 
As was noticed 40 years ago by Witten \cite{Wit8182}, for any ordinary supersymmetric system where the Hamiltonian represents the anticommutator $\{Q, \bar Q\}$ of complex supercharges, the supertrace
 \be
\lb{Witind}
I_W \ =\ \left \langle \left \langle e^{-\beta H} \right \rangle  \right \rangle
\ \stackrel{\rm def}= \  {\rm Tr} \left\{ (-1)^F  e^{-\beta H}  \right\}
 \ee
 is an integer number non depending on $\beta$. 

Indeed, all excited levels of the Hamiltonian are paired and only zero-energy states annihilated by the action of both $Q$ and $\bar Q$ may contribute. As a result, $I_W = 
n^{E=0}_B - n^{E=0}_F$ is an integer number not depending on $\beta$.

Take a 4-dimensional field theory. To regularize it in the infrared, one can put it in a finite toroidal spatial box. We obtain a supersymmetric quantum system with an infinite but discrete set of dynamical variables. It is important that the original supersymmetry algebra is kept intact under such regularisation. The Witten index \p{Witind} is still a $\beta$-independent integer.

A different situation arises when we put the theory not on $T^3$, but on $S^3$. In this case, the algebra is necessarily modified. Indeed, the momenta $P_a$, which commute on $\mathbb{R}^3$ and $T^3$, do not commute anymore: they realize now the isometries of $S^3$ and satisfy the $su(2)$ algebra. For an explicit construction, 
consider the embedding $x^2 + y^2 + z^2 + t^2 = \rho^2$ of $S^3$ of radius $\rho$ into $\mathbb{R}^4$. Consider then the operators $J^a = - \frac i2 \eta^a_{\mu\nu} x_\mu \partial_\nu$, where $\eta^a_{\mu\nu}$ are the  `t Hooft symbols \cite{Hooft},
\be
 J^1 &=& \frac i2 ( t\partial_x - x\partial_t + z\partial_y - y \partial_z), \nn
 J^2 &=& \frac i2 ( t\partial_y - y\partial_t + x\partial_z - z \partial_x), \nn
 J^3 &=& \frac i2 ( t\partial_z - z\partial_t + y\partial_x - x \partial_y).
\ee
Their commutators are $[J^a, J^b] = i \varepsilon^{abc} J^c$. At the vicinity of the north pole of $S^3$, $x_\mu = (\vec{0}, \rho)$, $J^a$ generate tangent space translations: $J^a \approx \frac {i\rho}2 \partial_a = - \frac {\rho}2 P_a$.

The supersymmetry may be broken completely (if, for example, one na\"ively replaces the flat Minkowski metric by a curved $S^3 \times \mathbb{R}$ metric in the Lagrangian), but, proceeding in a clever way and adding certain extra terms in the Lagrangian, one may keep a {\it part} of the original supersymmetry \cite{Sen}. The new algebra has the form

\be
\lb{alg-Sen}
[P_a, P_b]  &=& \frac {-2i}\rho \varepsilon_{abc} P_c, \nn
\ [Q_\alpha, P_a] &=& -\frac {1} {\rho} (\sigma_a Q)_\alpha, \qquad
 \ [ \bar Q^\alpha, P_a] =  \frac {1} {\rho} ( \bar Q \sigma_a)^\alpha, \nn
\{Q_\alpha, \bar Q^\beta\} &=& 2\left(H -\frac R \rho \right)\delta_\alpha{}^\beta \ + \ 2(\sigma_a)_\alpha^{\ \beta} P_a ,  \nn
 \ [ Q_\alpha, R]   &=&  -Q_\alpha, \qquad [\bar Q^\alpha, R] \ =\ \bar  Q^\alpha, \nn
\,[H, P_a] &=& [H, R] \ =\ [R,P_a] \ =\  0, \nn
\{Q_\alpha, Q_\beta\} &=& \{\bar Q^\alpha, \bar Q^\beta\} \ = \ [Q_\alpha, H] \ = \  [\bar Q^\alpha, H] \ = \ 0,
 \ee
where $\alpha = 1,2$ (bearing in mind further quantum-mechanical applications, we do not distiguish the dotted and undotted indices) and  $\sigma_a$ are the Pauli matrices.  
This algebra involves, besides $H, P_a, Q_\alpha, \bar Q^\alpha = (Q_\alpha)^\dagger$, also an extra $U(1)$ generator $R$. The presence of the latter is necessary: if one keeps the requirement $[Q_\alpha, H] = 0$, the Jacobi identities would not hold  without $R$ and the algebra would not be consistent. The mathematical notation for the algebra \p{alg-Sen} is\footnote{To be more precise, the algebra $su(2|1)$ includes only four bosonic generators: $H - R/\rho$ and $P_a$. We are dealing here with a {\it central extension} of this algebra.} $su(2|1)$. In the limit $\rho \to \infty$, the standard 4-dimensional ${\cal N} =1$ supersymmetry algebra  is reproduced. 

Not only the $4D$ supersymmetric field theories placed on $S^3 \times \mathbb{R}$ enjoy the algebra \p{alg-Sen}. This algebra also shows up in SQM systems not related to any field theory \cite{weak}. Following \cite{weak}, we will call {\it weak} this variety of supersymmetry.

The presence of the central charges in the anticommutator $\{Q_\alpha, \bar Q^\beta\}$ invalidates the usual claim that,  all the positive-energy states of the Hamiltonian are paired and there is an equal number of bosonic and fermionic degenerate states. As a result, the excited states may contribute to the supertrace \p{Witind} and the index is not an integer number anymore, but represents a nontrivial function of temperature. If the theory involves  charges ${\cal M}_j$ commuting with the Hamiltonian and the supercharges, one may introduce the associated chemical potentials $\mu_j$ and consider the supertraces $\langle \langle e^{-\beta H} e^{\mu_j {\cal M}_j} \rangle \rangle$. These supertraces may represent complicated functions of $\beta$ and $\mu_j$ \cite{Spir}, but the point is that they represent topological invariants in the same sense as the ordinary Witten index is --- they stay invariant under the deformations of the theory that keep the algebra intact. 

In the literature, this functional index is usually referred as {\it superconformal index}. That is how it was christened by its discoverers \cite{Romel1} and, indeed, this notion is very useful for studying the dynamics of superconformal theories
\cite{Rastelli}. But this index has more general scope. First of all, a $4D$ theory to be placed on  $S^3 \times \mathbb{R}$ need not necessarily be conformal. Second, as we mentioned, there exist  supersymmetric system enjoying the algebra 
isomorphic to \p{alg-Sen} and not related to any field theory. That is why we will not use the word superconformal, but will either call this index  "Witten index for weak SUSY systems", as in the title, or {\it R\"omelsberger's index}.

The main original point of this paper is a simple proof of the invariance of R\"omelsberger's index under deformations. Then we will illustrate this general theorem by two examples: {\it (i)} the simplest weak supersymmetry model of Ref. \cite{weak} and {\it (ii)} R\"omelsberger's model --- a weak supersymmetric quantum mechanical model involving a complex bosonic dynamical variable and arising when the massless Wess-Zumino $4D$ model is put on  $S^3 \times \mathbb{R}$ and the higher spherical harmonics of the fields are suppressed.

\section{Invariance of the index}
\setcounter{equation}0

\begin{thm} Let $Q_\alpha, \bar Q^\alpha$ and $H = P_0$ but not $P_a$ and $R$ be functions of parameter $\gamma$ such that the algebra \p{alg-Sen} stays intact. Then
 \be
\lb{theor}
\frac d{d\gamma} 
 \left \langle \left \langle e^{-\beta H} \right \rangle  \right \rangle \ =\ 0.
\ee
\end{thm}
\begin{proof}

By expanding $e^{-\beta H}$ into the series and using the cyclic property of the supertrace, we deduce
\be 
\lb{diff-exp}
\frac d{d\gamma} 
 \left \langle \left \langle e^{-\beta H} \right \rangle  \right \rangle \ =\ - \beta 
\left \langle \left \langle \frac {dH}{d\gamma} e^{-\beta H} \right \rangle  \right \rangle.
 \ee
The third line in \p{alg-Sen} reads
$$ \{Q_\alpha, \bar Q^\beta\} \ =\  2 H \delta_\alpha{}^\beta \quad+ \quad {\rm central \ charges}. $$
Capitalizing on the assumed $\gamma$-independence of the central charges, we deduce
\be
\lb{dH-dQ}
 \frac {dH}{d\gamma} \ =\  \frac 14 \left\{Q_\alpha, \frac {d \bar Q^\alpha}{d\gamma} \right \} + 
\frac 14 \left\{\frac {dQ_\alpha}{d\gamma}, \bar Q^\alpha \right\}.
\ee
\begin{lem}
\be
\lb{lem}
\left \langle \left \langle \{Q_\alpha, V\} e^{-\beta H} \right \rangle  \right \rangle \ =\ 0
 \ee
for any (not too wild) $V$.
\end{lem}
\begin{proof}
Take for definiteness $\alpha =1$. By definition,
\be  
\left \langle \left \langle O \right \rangle  \right \rangle \ =\ \sum_B \langle B|O|B \rangle - \sum_F \langle F|O|F\rangle,
 \ee
where $|B\rangle$ and $|F\rangle$ are the bosonic and fermionic states. 

The weak SUSY algebra \p{alg-Sen} includes the ordinary ${\cal N} = 2$ SQM subalgebra ${\cal A}_1$ with the generators $Q_1, \bar Q_1$ and 
\be
H_1 = \frac 12  \{Q_1, \bar Q^1\} \ =\ H +P_3 - \frac R \rho.
 \ee 
We choose the eigenstates of $H_1$ (which are also the eigenstates of $H$ due to  $[H_1, H] = 0$) as the basis in Hilbert space. From the viewpoint of ${\cal A}_1$, the spectrum includes:

{\it i)} The states annihilated by the action of both $Q_1$ and $\bar Q^1$. If the symmetry ${\cal A}_1$  is not broken spontaneously, these are the ground states of $H_1$, 

{\it ii)} The doublets $(B,F)$ satisfying 
 \be
\lb{matr-el}
 Q_1 |B\rangle = |F\rangle, \quad Q_1|F\rangle = 0, \quad \langle B|Q_1 = 0, \quad \langle F| Q_1 = \langle B| . 
 \ee

Using $[Q_1, H] = 0$, we may rewrite \p{lem} as 
\be
\left \langle \left \langle \{Q_1, V\} e^{-\beta H} \right \rangle  \right \rangle \ =\  \left \langle \left \langle Q_1 V e^{-\beta H} \right \rangle  \right \rangle + \left \langle \left \langle V  e^{-\beta H} Q_1\right \rangle  \right \rangle \nn
= \ \sum_B \langle B | Q_1 V e^{-\beta H} + V  e^{-\beta H} Q_1 |B \rangle - 
\sum_F \langle F | Q_1 V e^{-\beta H} + V  e^{-\beta H} Q_1 |F \rangle.
 \ee
It is immediately seen that the singlet states to not contribute.
Next, using \p{matr-el}, we obtain
 \be
\left \langle \left \langle \{Q_1, V\} e^{-\beta H} \right \rangle  \right \rangle \ =\  \sum_{\rm doublets}
\langle B| V e^{-\beta H} |F \rangle - \sum_{\rm doublets}
\langle B| V e^{-\beta H} |F \rangle \ =\ 0.
\ee
By the same token we can prove 
\be
\lb{lem-bar}
\left \langle \left \langle \{\bar Q^\alpha, V\} e^{-\beta H} \right \rangle  \right \rangle \ =\ 0
 \ee
for any $V$.
\end{proof}
By combining \p{diff-exp}, \p{dH-dQ}, \p{lem}, \p{lem-bar}, we arrive at \p{theor}.
\end{proof}

{\bf Remarks}. 

\begin{enumerate}

\item The whole reasoning above also works for a {\it generalized} index 
\be 
\lb{ind-muM}
\tilde I \ =\  \left \langle \left \langle e^{-\beta H} e^{\mu M}\right \rangle  \right \rangle\,
 \ee
where $M$ is an operator that commutes with the Hamiltonian and at least one pair of the supercharges. One can take for $M$ the operator $P_3 - R/\rho$, which is present in the algebra \p{alg-Sen} and  commutes with  $Q_1$ and $\bar Q^1$ (the corresponding index is reduced to \p{Witind} depending on a certain combination of $\beta$ and $\mu$), but a rich enough dynamical system may involve many such extra integrals of motion $M_j$.

Then the  index 
\be 
\lb{ind-muMj}
I(\beta, \mu_j) \ =\  \left \langle \left \langle e^{-\beta  H} e^{\mu_j M_j}\right \rangle  \right \rangle\,
 \ee
 is  invariant under the deformations described above.

\item The theorem just proven represents a particular case of the so-called {\it equivariant index theorem} known to mathematicians (see Ref. \cite{Pestun} for a review addressed to physicists). The notion of the equivariant index was first introduced back in 1950 by Cartan \cite{Cartan}. 

\end{enumerate}

\section{Weak supersymmetric harmonic oscillator and its deformation}
\setcounter{equation}0
 Being confronted with a complicated physical problem, Enrico Fermi used to ask his collaborators and himself --- what is the hydrogen atom for this problem?
In our case, we have it. The simplest weak supersymmetric system \cite{weak} is even much simpler than the hydrogen atom: it includes only one real degree of freedom and the oscillator-like Hamiltonian, \footnote{Classically, $\bar \psi^\alpha = (\psi_\alpha)^\dagger$; the corresponding quantum operator is $\bar \psi^\alpha = \partial/\partial \psi_\alpha$; the indices are raised and lowered by $\varepsilon^{\alpha\beta} = - \varepsilon_{\alpha\beta}$ with the convention $\varepsilon^{12} = 1$.} 
 \be
\lb{Ham-real}
H = \frac {p^2 + x^2}2 + \frac 12 (\bar \psi^\alpha  \bar\psi_\alpha + \psi_\alpha \psi^\alpha).
 \ee
This Hamiltonian commutes with the supercharges\footnote{We changed the notation $S_\alpha \to Q_\alpha \sqrt{2}$, compared to Ref. \cite{weak}.}
\be
\lb{Q-real}
Q_\alpha = \ (p-ix)\,\frac{\psi_\alpha - \bar \psi_\alpha}{\sqrt{2}}, \qquad \bar Q^\alpha = (p+ix)\,\frac{\psi^\alpha + \bar \psi^\alpha}{\sqrt{2}}.
\ee 
The full algebra includes also the central charges $\in su(2) \oplus u(1)$ (note that $Z_{\alpha\beta} = Z_{\beta\alpha}$ and $Z_\alpha{}^\alpha = 0$),
\be
\lb{ZY}
Z_\alpha{}^\beta \ =\ \psi_\alpha \bar \psi^\beta + \psi^\beta \bar \psi_\alpha, \qquad Y = \frac 12 (\bar \psi^\alpha  \bar \psi_\alpha + \psi_\alpha \psi^\alpha)
 \ee
and reads
 \be
\lb{alg-real}
\{ Q_\alpha, Q_\beta\} &=& \{ \bar Q^\alpha, \bar Q^\beta\} \ =\ 0, \nn
\{ Q_\alpha, \bar Q^\beta\} &=& (2H - Y)\delta_\alpha{}^\beta +  Z_\alpha{}^\beta, \nn
\, [Q_\alpha, Z_\beta{}^\gamma] &=& \varepsilon_{\alpha\beta}\, Q^\gamma  - \delta_\alpha{}^\gamma \, Q_\beta, 
\qquad [\bar Q^\alpha, Z_\beta{}^\gamma] \ = \ \delta_\beta{}^\alpha \,\bar Q^\gamma -  \varepsilon^{\alpha\gamma} \,\bar Q_\beta , \nn
\,[Q_\alpha, Y] &=& - Q_\alpha, \qquad [\bar Q^\alpha, Y] \ =\   \bar Q^\alpha, \nn
 \,[Z_\alpha{}^\beta, Z_\gamma{}^\delta] &=& \delta_\gamma{}^\beta \,Z_\alpha{}^\delta -
\delta_\alpha{}^\delta \,Z_\gamma{}^\beta+
 \varepsilon_{\alpha\gamma} \,Z^{\beta\delta} - \varepsilon^{\beta\delta} \, Z_{\alpha\gamma}, \nn
\,[H, Z_\alpha{}^\beta] &=& [H,Y] \ = \ [Y, Z_\alpha{}^\beta] \ =\ 0.
 \ee
It is isomorphic to \p{alg-Sen}, as one can be explicitly convinced by replacing in \p{alg-real}
\be
H \ \to \ \frac {H\rho}{\sqrt{2}}, \quad Q_\alpha \ \to \ Q_\alpha \sqrt{\rho}, \quad Y \ \to \ R, \quad  
Z_\alpha{}^\beta \ \to \ \rho P_a (\sigma_a)_\alpha{}^\beta
  \ee

 and using the identities
\be
 (\sigma_a)_\alpha{}^\gamma (\sigma_a)_\beta{}^\delta &=& \delta_\alpha^{\ \delta}  \delta_\beta^{\ \gamma} 
+ \varepsilon_{\alpha\beta} \varepsilon^{\gamma\delta}, \nn
\varepsilon_{abc} (\sigma_a)_\alpha{}^\beta (\sigma_b)_\gamma{}^\delta &=& 
\frac i2 \left[\delta_\gamma^{\ \beta} (\sigma_c)_\alpha{}^\delta -  \delta_\alpha^{\ \delta} (\sigma_c)_\gamma{}^\beta
+ \varepsilon_{\alpha\gamma} (\sigma_c)^{\beta\delta} - \varepsilon^{\beta\delta} (\sigma_c)_{\alpha\gamma}\right].
  \ee

 \begin{figure} [ht!]
      \bc
    \includegraphics[width=.6\textwidth]{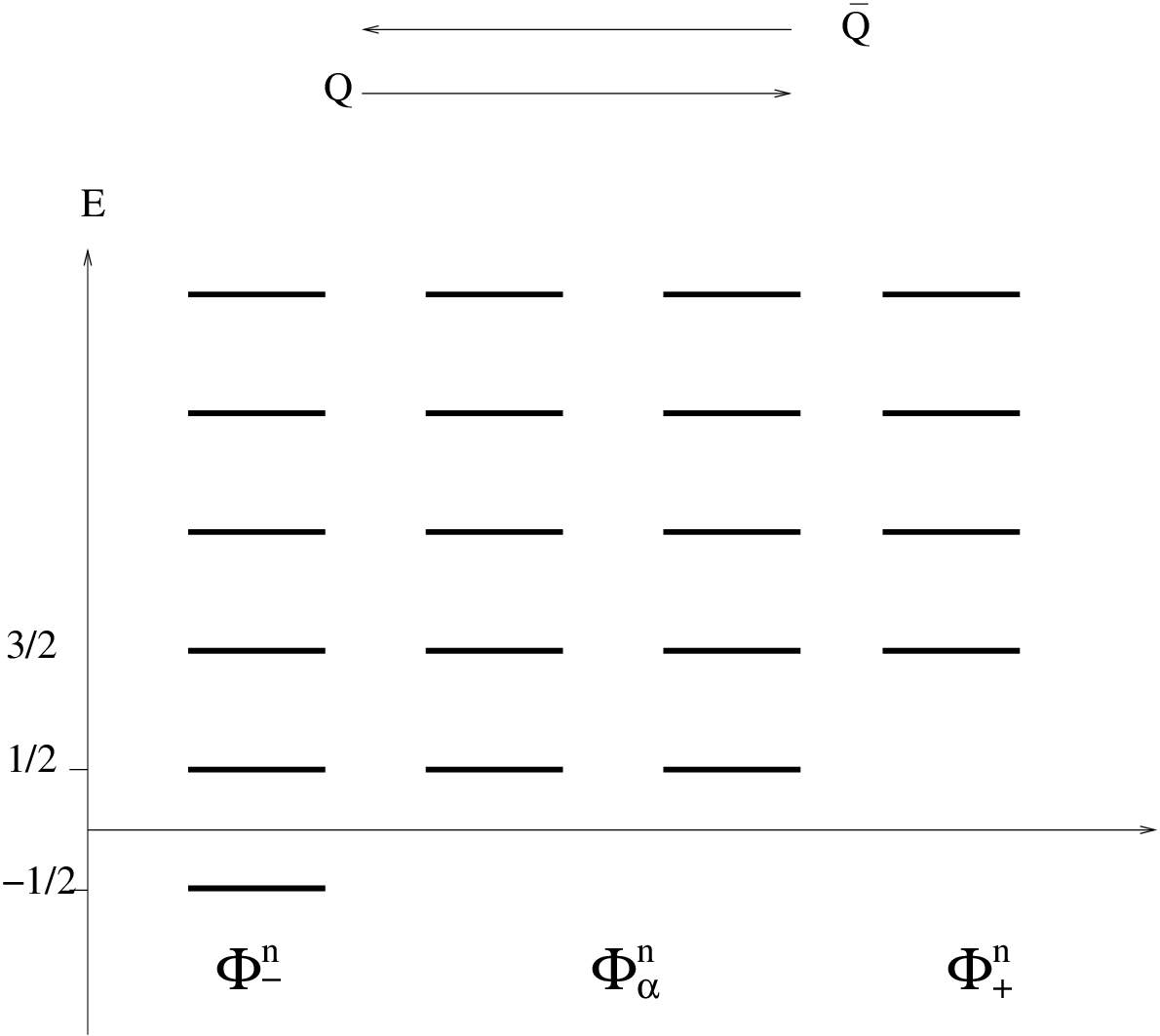}                  
     \ec
    \caption{The  \p{Ham-real}}        
 \label{spec-real}
    \end{figure}  

The spectrum of Hamiltonian \p{Ham-real} depicted in Fig. \ref{spec-real} consists of 4 towers. The left tower involves the states 
$\Phi_-^n  =  |n\rangle (1- \psi_1\psi_2)$, where $|n\rangle$ are the ordinary bosonic oscillator states. The right tower 
 involves the states 
$\Phi_+^n  =  |n\rangle (1+ \psi_1\psi_2)$ and two central fermion towers involve the states $\Phi_\alpha^n = 
|n\rangle \psi_\alpha$. In contrast to the ordinary supersymmetic system where the ground state has a zero or maybe positive energy (if supersymmetry is spontaneously broken),  the ground state $\Phi_-^0$  of \p{Ham-real} annihilated by the action of $Q_\alpha$ and $\bar Q^\alpha$
has the energy $E = -1/2$.  The violation of the familiar supersymmetry pattern is, of course, related to the presence of the central charges in $\{ Q_\alpha, \bar Q^\beta\}$. If one wishes, one can add a constant and bring the ground state energy to zero,  but there is no impelling reason to do so.

At the first excited level, we have {\it three} states of energy $E = 1/2$; the boson and the fermion states are not paired.
This happens because  $Q_\alpha|\Phi^0_\beta\rangle = 0$ and the fourth component of the multiplet is missing. For the energies $E=3/2, 5/2, \ldots$, we have $Q_1 Q_2 |\Phi^{n\geq 2}_-\rangle \, \neq 0$, giving rise to habitual degenerate quartets. 

The index \p{Witind} has two contributions: from the ground state and from the triplet of first excitations. It reads
 \be
\lb{IW-real}
I \ =\ e^{\beta/2} - e^{-\beta/2} \ =\ 2\sinh (\beta/2).
 \ee
An alternative interpretation of the result \p{IW-real} is the following. Represent the index as 
\be
\lb{ind-fugac}
I \ =\ \left \langle \left \langle e^{-\beta H_1} e^{\beta(Z_1^{\,1} - Y)/2} \right \rangle \right \rangle,
 \ee
where $H_1 = \{Q_1, \bar Q^1\}/2$.
From the viewpoint of ${\cal A}_1$, we are calculating the generalized Witten index including the fugacities
$e^{\beta Z_1^{\,1}/2}$ and $e^{-\beta Y/2}$ associated with the conserved charges $Z_1^{\,1}$ and $Y$. 
Only the vacuum zero energy states of $H_1$ --- the state $(1-\psi_1\psi_2) |0 \rangle$ and the state $\psi_2 |0 \rangle$
contribute in \p{ind-fugac}. The different contributions of the two states is now attributed not to their different energies, but to different fugacities.

\vspace{1mm}

The system \p{Ham-real} is the simplest weak supersymmetric system, but it can be deformed to include nontrivial interactions, while keeping the algebra \p{alg-real} intact. The deformed quantum supercharges and the Hamiltonian read \cite{weak} \footnote{This system and also the weak supersymmetric system with complex field to be discussed later  admit a  superfield description worked out in \cite{IvSid}. }
  \be
\lb{Q-deform}
\sqrt{2}\, Q_\alpha &=& [p - iV(x)]( \psi_\alpha - \bar \psi_\alpha) + iB(x)\left(\bar \psi^\beta \psi_\alpha \bar \psi_\beta -  \psi_\beta \bar \psi_\alpha \psi^\beta   \right), \nn
\sqrt{2} \, \bar Q^{\alpha} &=& [p + iV(x)](\psi^\alpha + \bar \psi^\alpha) - iB(x) \left( \bar \psi^\beta  \psi^\alpha \bar \psi_\beta + \psi_\beta \bar \psi^\alpha \psi^\beta \right),
\ee
\be
\lb{Hreal-deform}
H \ =\ \frac {p^2}2 + \frac {V(x)^2}2 + \frac {V'(x)} 2  (\psi_\alpha \psi^\alpha + \bar \psi^\alpha \bar \psi_\alpha) + \frac {B'(x)}2 (\psi_\alpha \psi^\alpha \bar \psi^\beta \bar \psi_\beta - 2 \psi_\alpha \bar \psi^\alpha + 1) + \frac {B^2}2\,,
 \ee
where $V(x)$ is an arbitrary function and 
$$
B(x) \ =\ \frac {V'(x) - 1}{2V(x)}.
$$
The central charges $Z_\alpha{}^\beta$ and $Y$  are {\it not} deformed, they have the same form \p{ZY} as before.

As was explicitly shown in Ref. \cite{weak}, the ground state and the triplet of first excitations are not shifted under the deformation. In fact, the system \p{Hreal-deform} is related to certain {\it quasi-exactly solbavle} models, where the energies of the ground state and the first excitation are  fixed \cite{quasi}.
 As a result, the index is still given by the function \p{IW-real}.

The invariance of the index also follows, of course, from the general theorem proven in the previous section.

\section{Complex model}. 
\setcounter{equation}0

One of the ways to put the theory of  free massless chiral superfield $\Phi = \phi(x^\mu_L) + \sqrt{2} \,\theta_\alpha \psi^\alpha (x^\mu_L)$ on $S^3$ of radius $\rho$ while keeping as much of the original supersymmetries as possible is to write the Hamiltonian

\be
\lb{Ham-complex-gen}
H \ =\  \frac 1{\rho^3} \bar \pi  \pi + \rho \,  {\bar \phi  \phi}  + \frac {\rho^2} 2 (\psi_\alpha \bar \psi^{\alpha} - \bar \psi^{\alpha} \psi_\alpha) \ + \cdots\ ,
 \ee
where  $\phi$ etc. are the {\it constant} field modes, the dots standing for the contribution of all other modes, which we will disregard. The canonical (anti)commutators are 
$$[\pi, \phi] = [\bar \pi, \bar \phi] \ =\  -i, \qquad \{\psi_\alpha,  \bar \psi^\beta\} \ = \ 
\frac 1{\rho^3} \delta_\alpha^{\ \beta}.$$
 The Hamiltonian \p{Ham-complex-gen} commutes with the supercharges 
 \be
Q_\alpha \ =\ \sqrt{2}\, \psi_\alpha (\pi - i \rho^2 \bar \phi), \qquad \bar Q^\alpha \ =\ \sqrt{2} \, \bar \psi^\alpha  (\bar \pi + i \rho^2 \phi) .
 \ee
Now, $\{Q_\alpha, Q_\beta\} = 0$, but the  anticommutator $\{Q_\alpha, \bar Q^\beta\}$ involves the familiar central charges:
\be
\lb{alg-compl-gen}
\{Q_\alpha, \bar Q^\beta\} \ =\ 2\left(H+ \frac L\rho \right) \delta_\alpha{}^\beta + \frac 2\rho Z_\alpha{}^\beta,
 \ee
where 
 \be
\lb{LZ}
L \ =\ i(\phi \pi - \bar \phi \bar \pi)
 \ee
 has the meaning of the angular momentum in the $\phi$ plane and $ Z_\alpha{}^\beta$ was defined in \p{ZY}. The commutator  $[Z_\alpha{}^\beta, Z_\gamma{}^\delta]$ was written in \p{alg-real} and 
the other nontrivial commutators are 
\be
&&[Q_\alpha, L]\  =\  Q_\alpha, \qquad  [\bar Q^\alpha, L] \ =\  -\bar Q^\alpha, \nn
&& [Q_\alpha, Z_\beta{}^\gamma] \ =\  
\varepsilon_{\alpha\beta} Q^\gamma  - \delta_\alpha{}^\gamma Q_\beta, \qquad [\bar Q^\alpha, Z_{\beta\gamma}] \ =\  
\delta_\beta{}^\alpha  \bar Q^\gamma - \varepsilon^{\alpha\gamma}  \bar Q_\beta\,.
 \ee
This algebra is isomorphic to \p{alg-Sen} and \p{alg-real}.

In the following, we set  for simplicity $\rho=1$, which gives
\be
\lb{Ham-complex}
H \ =\   \bar \pi  \pi +  {\bar \phi  \phi}  + \frac {1} 2 (\psi_\alpha \bar \psi^{\alpha} - \bar \psi^{\alpha} \psi_\alpha),
\qquad Q_\alpha \ =\  \sqrt{2}\, \psi_\alpha  (\pi - i  \bar \phi)
 \ee
and 
\be
\lb{alg-compl}
\{Q_\alpha, \bar Q^\beta\} \ =\ 2\left(H+  L \right) \delta_\alpha{}^\beta +  2 Z_\alpha{}^\beta,
 \ee

The spectrum of $H$ is the spectrum of the 2-dimensional oscillator shifted by  $F-1$, where $F$ is  the fermion charge --- the eigenvalue of the operator $\psi_\alpha \bar \psi^\alpha$. We have
 \be
\lb{spec-H0}
E^F_{nl} \ =\ 2n + |l| +F,
 \ee
where $n$ is the number of the radial excitation.
Note now that the supercharges also commute with the operator $K = L + F$. We can thus consider a modified Hamiltonian 
\be
\lb{Ham-lam}
H_\lambda = H +\lambda K.
\ee
The bosonic part of this Hamiltonian describes a 2-dimensional oscillator supplemented by a magnetic field. 
The anticommutator $\{Q_\alpha, \bar Q^\beta\}$ is also modified:
\be
\lb{anti-mod}
\{Q_\alpha, \bar Q^\beta\} \ =\ 2\delta_\alpha{}^\beta (H_\lambda - \lambda K + L) + 2Z_\alpha{}^\beta.
 \ee
That does not mean, of course, an essential modification of the algebra, it is still basically $su(2|1)$. But the anticommutator $\{Q_\alpha, \bar Q^\beta\}$, being expressed in terms of $H_\lambda$, does not have the same functional form as \p{alg-compl} and includes the extra operator $F$.  The theorem of Sect. 2 does not apply in this case,  the index \p{Witind} of $H_\lambda$  need not be the same as the index of $H$, and it is not.

The spectrum of $H_\lambda$ reads 
\be
\lb{spec-Hlam}
E^F_{nl} \ =\ 2n + |l| +F +\lambda(l+F).
 \ee 

 If we want the spectrum to be bounded from below, the  parameter $\lambda$ should not exceed unity. 

Well, in principle, there is nothing wrong with the {\it free} Hamiltonian with $\lambda >1$. It would represent an  example of the system with {\it benign} ghosts \cite{ghosts}. But we are interested in this paper with nonlinear  deformations keeping the weak supersymmetry algebra intact. If $\lambda > 1$, such a deformation would make the ghosts {\it malignant} and unitarity of the theory would be violated.

As was mentioned, the Hamiltonians \p{Ham-complex} and \p{Ham-lam} are related to the theory of free massless chiral multiplet placed on $S^3 \times \mathbb{R}$. Suppose that we want to put there a theory involving a superpotential $W(\Phi)$ in such a way that the weak SUSY algebra is kept intact. One can then derive that \cite{Romel2,IvSid}

{\it (i)} It is only possible for a superpotential $W(\Phi) \propto \Phi^n$.

{\it (ii)} The value of $\lambda$ must take a particular value 
\be
\lb{lam-n}
 \lambda = 1 - 2/n.
 \ee

As far as algebraic properties are concerned, $n$ needs not be integer. But we will be interested only in the models with integer $n$, where the potential does not involve ugly branchings at the origin. In the first place, in the renormalizable Wess-Zumino model with $n=3$.  

For a generic $\lambda$, the  spectrum \p{spec-Hlam} has a complicated structure. It involves an infinite number of ``castles", each consisting of four towers, as in Fig.\ref{spec-real}. The first such castle grows from the ``cellar" $\Psi^0_{00}$
with zero energy. At its ground floor, we find a bosonic state $\Psi^0_{01}$ with energy $E = 1 + \lambda$ and a couple of fermionic states $\Psi^1_{00}$ with the same energy. Higher floors of this castle represent degenerate quartets. At the first floor with energy $E = 2(1 + \lambda)$, we find the bosonic states $\Psi^0_{02}$ and   $\Psi^2_{20}$ and a couple of fermionic states  $\Psi^1_{01}$.

The cellar of the second castle is the state  $\Psi^0_{0,-1}$ with energy $E = 1-\lambda$. The ground floor has the energy $E = 2$ and includes the bosonic state  $\Psi^0_{02}$ and two fermionic states  $\Psi^1_{0,-1}$. At the first floor with  energy $E = 3 + \lambda$, we have the state  $\Psi^0_{11}$,  two  states $\Psi^1_{10}$ and  the state $\Psi^2_{0,-1}$. And so on.  
The higher castles grow from the states  $\Psi^0_{0,-m}$ with energy $E = m(1-\lambda)$. The gap between the floors in all the castles is $\Delta E = 1 + \lambda$.

The index \p{Witind} acquires the contributions from the cellars and ground floors in each castle. The calculation gives

\be
\lb{ind-lam}
I(\lambda) \ =\ \left[ 1 - e^{-\beta(1 + \lambda)} \right] \left[ 1 + e^{-\beta(1 - \lambda)} + e^{-2\beta(1 - \lambda)} + \cdots \right] \ =\ \frac {1 - e^{-\beta(1 + \lambda)}}{1 - e^{-\beta(1 - \lambda)}},
\ee
in agreement with Eq.(31) in the last R\"omelsberger's paper \cite{Romel2}.\footnote{Note the difference in notations. R\"omelsberger's $\Xi$ is the same as our $H$.} 

However, the spectrum simplifies a lot for the special values of $\lambda$ in \p{lam-n}. Take $n = 3$ giving $\lambda = 1/3$. Then the triplet on the  ground floor of the $m$-th castle becomes degenerate with the cellar of the castle $m+2$.
As a result, the spectrum now involves two bosonic singlets with energies $E=0$ and $E = 2/3$, while all the other states belong to the supersymmetric quartets. The index is
\be
\lb{ind-n3}
I(3) \ =\ 1 + e^{-2\beta/3}.
 \ee

The extra degeneracies appear also for higher $n$ when  the triplet on the  ground floor of the $m$-th castle becomes degenerate with the cellar of the castle $m+n-1$. We are left with $n-1$ singlet bosonic states, the other states are in the quartets, and the index is 
\be 
\lb{ind-n}
I(n) \ =\ \sum_{m=0}^{n-2} e^{-2\beta m/n}.
 \ee
The results \p{ind-n3} and \p{ind-n} [which also follow from \p{ind-lam}] are quite natural. If the Wess-Zumino model with superpotential $W(\Phi) \propto \Phi^n$ is placed on $T^3 \times \mathbb{R}$, the Witten index is known to be equal to $n-1$. But on $S^3$ the degenerate vacuum states are equidistantly split with the gap $\Delta E = 2/n$ [actually, 
$\Delta E = 2/(n \rho)$], and the index becomes a function of the ratio $\beta/\rho$.

We go back to the simplest case $n=3$. The deformed supercharges and the Hamiltonian read
\be
\lb{Qdeform-compl}
&& Q_\alpha \ =\ \sqrt{2} \left[(\pi - i\bar \phi) \psi_\alpha + i\gamma \bar \psi_\alpha \bar \phi^2 \right], \nn
&& \bar Q^\alpha \ =\ \sqrt{2} \left[ (\bar\pi + i \phi) \bar \psi^\beta + i\gamma  \psi^\alpha \phi^2 \right],
  \ee
\be
\lb{Hdeform-n3}
H \ =\ \bar \pi \pi + \bar \phi \phi + \psi_\alpha \bar \psi^\alpha -1 \ + \ \frac 13 [i(\phi \pi - \bar \phi \bar \pi) + \psi_\alpha \bar \psi^\alpha] \nn
+ 2\gamma (\psi_1 \psi_2 \phi + \bar \psi^2 \bar \psi^1 \bar \phi) + \gamma^2 (\bar \phi \phi)^2,
 \ee
where $\gamma$ is the deformation parameter.

They satisfy the same algebra   as in the undeformed case:
\be
\lb{alg-gam}
\{Q_\alpha, \bar Q^\beta\} \ =\ 2\delta_\alpha{}^\beta \left(H  + \frac {2L}3  - \frac F3\right) + 2Z_\alpha{}^\beta.
 \ee
By the theorem proven in Sect. 2, the index stays invariant under such deformation.

\vspace{1mm}

This implies the following properties of the deformed spectrum:

\begin{itemize}
\item The energies of the singlet bosonic states are still $E_0 = 0$ and $E_1 = 2/3$.

\item The  quartets stay degenerate, but their position may be shifted.
 \end{itemize}

These conclusions were confirmed by an explicit perturbative calculation to the order $\gamma^2$. 
The singlets are not shifted and the quartets are. For example a ``new" quartet (that appears in the spectrum for a special value $\lambda = 1/3$ and has the energy $E = 2m/3$ in the undeformed case, involving the states $\Psi^0_{0,-2}$, $2\Psi^1_{00}$ and $\Psi^0_{01}$ for $m=2$ and the states $\Psi^0_{0,-m}$, $2\Psi^1_{0,2-m}$ and $\Psi^0_{1,3-m}$ for $m\geq 3$) is shifted up by
 \be
\delta E_\gamma \ =\ \frac {m(m-1)}4 \, \gamma^2.
 \ee
 
Note that the undeformed model also has  extra degeneracies which {\it are} lifted under deformation. For example, the ``new" quartet with the states $\{\Psi^0_{0,-4}, 2\Psi^1_{0,-2}, \Psi^0_{1,-1}\}$ has the same energy $E = 8/3$ as the ``old" quartet  (the first floor of the first castle) with the states   $\{\Psi^0_{02}, 2\Psi^1_{01}, \Psi^2_{00}\}$ before the deformation. After the deformation, the new quartet is shifted up by $3\gamma^2$ and the old one  only by $3\gamma^2/2$.

    \section*{Akcnowledgements}
    
    I am indebted to  to Sergei Fedoruk and Stepan Sidorov for useful comments and to Bruno Le Floch, Vasily Pestun and  Vyacheslav Spiridonov  for illuminating discussions and valuable remarks.

  \section*{Appendix}
For a dedicated reader who might wish to check our calculations for the complex model, we present here the explicit expressions for first few normalized eigenstates of the Hamiltonian $H_{1/3}$.
\vspace{1mm}

$\bullet E=0$:
$$ \Psi^0_{00} \ =\ \frac 1{\sqrt{\pi}} \, e^{-\phi \bar \phi},$$

$\bullet E=2/3$:
$$ \Psi^0_{0,-1} \ =\ \sqrt{\frac 2 \pi} \, \bar \phi \, e^{-\phi \bar \phi},$$
 
$\bullet E=4/3$:
$$ \Psi^0_{01} \ =\ \sqrt{\frac 2 \pi} \,  \phi \,  e^{-\phi \bar \phi}, \qquad \Psi^0_{0,-2} = \sqrt{\frac 2 \pi} \,  \bar \phi^2\, e^{-\phi \bar \phi}, \qquad \Psi^1_{00,\alpha} = \frac {\psi_\alpha}{\sqrt{\pi}} \,e^{-\phi \bar \phi},$$

 $\bullet E=2$:
$$ \Psi^0_{10} \ =\ \frac 1{\sqrt{\pi}}(1 - 2\phi \bar \phi)\, e^{-\phi \bar \phi}, \qquad \Psi^0_{0,-3}\ =\ \sqrt{\frac 4
{3 \pi}}\,\bar \phi^3 \, e^{-\phi \bar \phi}, \qquad \Psi^1_{0,-1,\alpha} = \ \psi_\alpha \sqrt{\frac 2{\pi}} \,  \bar \phi \, e^{-\phi \bar \phi}, $$

etc.

\end{document}